\documentclass[twoside,leqno,twocolumn]{article}  
\usepackage{ltexpprt}

\usepackage{amsmath,amssymb}
\usepackage{graphicx}
\usepackage{psfrag}

\def\nfrac#1#2{{\textstyle\frac{#1}{#2}}}
\def\dfrac#1#2{\lower0.15ex\hbox{\large$\frac{#1}{#2}$}}

\def \qedbox{\hfill\vbox{\hrule\hbox{\vrule height1.3ex\hskip0.8ex\vrule}\hrule}}

\begin{document}

\title{\Large{The switch Markov chain for sampling irregular graphs}\\ 
        {\large (Extended Abstract)}}
\author{Catherine Greenhill\thanks{{\tt c.greenhill@unsw.edu.au},\
School of Mathematics and Statistics, UNSW Australia, Sydney 2052, Australia.
Research supported by Australian Research Council Discovery Project DP140101519.}}

\date{}

\maketitle

\begin{abstract}
The problem of efficiently sampling from a set of (undirected) graphs 
with a given degree sequence has many applications. One approach to this problem uses
a simple Markov chain, which we call the switch chain, to perform the sampling.  
The switch chain is known to be rapidly mixing for regular degree sequences.
We prove that the switch chain is rapidly mixing for any
degree sequence with minimum degree at least 1 and with maximum degree $d_{\max}$
which satisfies $3\leq d_{\max}\leq \nfrac{1}{4}\, \sqrt{M}$, where $M$ is the 
sum of the degrees. 
The mixing time bound obtained is only an order of $n$ larger than that
established in the regular case, where $n$ is the number of vertices.
\end{abstract}

\section{Introduction}\label{s:intro}

The switch chain is a natural Markov chain for sampling from a set of graphs
with a given degree sequence.
Each move of the chain selects two distinct, non-incident edges
edges uniformly at random and attempts to replace these edges by  
a perfect matching of the four endvertices, chosen
uniformly at random.  The proposed move is rejected if a multiple edge
would be formed.

We call each such move a \emph{switch}.  Ryser~\cite{ryser} used switches to study the
structure of 0-1 matrices. Markov chains based on switches have been introduced by 
Besag and Clifford~\cite{BC89} for 0-1 matrices (bipartite graphs), Diaconis and Sturmfels~\cite{DS98}
for contingency tables and Rao, Jana and Bandyopadhyay~\cite{RJB} for directed graphs.

The switch chain is aperiodic and its transition matrix is symmetric. 
It is well-known that the switch chain is irreducible for any (undirected)
degree sequence: see~\cite{petersen,taylor}. 

In order for the switch chain to be useful for sampling, it must converge quickly
to its stationary distribution.  
(For Markov chain definitions not given here, see~\cite{Jerrum}.)

Cooper, Dyer and Greenhill~\cite{CDG,CDG-corrigendum}
showed that the switch chain is rapidly mixing for regular undirected graphs.
Here the degree $d=d(n)$ may depend on $n$, the number of vertices.  The mixing time
bound is given as a polynomial in $d$ and $n$. 
Earlier, Kannan, Tetali and Vempala~\cite{KTV} investigated the mixing time 
of the switch
chain for regular bipartite graphs.  
Greenhill~\cite{directed} proved that the
switch chain for regular directed graphs (that is, $d$-in, $d$-out directed graphs)
is rapidly mixing, again for any $d=d(n)$.
Mikl{\' o}s, Erd{\H o}s and Soukup~\cite{MES} proved that the switch chain is
rapidly mixing on half-regular bipartite graphs; that is, bipartite degree sequences
which are regular for vertices on one side of the bipartition, but need not be
regular for the other.

The proofs of all these mixing results used a multicommodity flow argument~\cite{sinclair}. 
In each case, regularity (or half-regularity) was only required for one 
lemma, which we will call the \emph{critical lemma}.  This is a counting
lemma which is used to bound
the maximum load of the flow 
(see~\cite[Lemma 4]{CDG},~\cite[Lemma 5.6]{directed} and~\cite[Lemma 6.15]{MES}).  

In Section~\ref{s:analysis} we give an alternative proof of the critical lemma
which does not require regularity.  This
establishes the following theorem, extending the 
rapid mixing result from~\cite{CDG} to irregular degree sequences
which are not too dense.  


Given a degree sequence $\boldsymbol{d} = (d_1,\ldots, d_n)$, write
$\Omega(\boldsymbol{d})$ for the set of all (simple) graphs with vertex
set $[n] = \{ 1,2,\ldots, n\}$ and degree sequence $\boldsymbol{d}$.
Recall that $\boldsymbol{d}$ is called \emph{graphical} when $\Omega(\boldsymbol{d})$
is nonempty.  We restrict our attention to graphical sequences.
Write $d_{\min}$ and $d_{\max}$ for the minimum and maximum degree in $\boldsymbol{d}$,
respectively, and let $M = \sum_{j=1}^n d_j$ be the sum of the degrees.

\begin{theorem}
Let $\boldsymbol{d} = (d_1,\ldots, d_n)$ be a graphical degree sequence 
such that $d_{\min}\geq 1$ and 
\[ 3\leq d_{\max}\leq \frac{1}{4}\, \sqrt{M}.\]
The mixing time $\tau(\varepsilon)$ of the switch Markov chain with
state space $\Omega(\boldsymbol{d})$ satisfies
\[ \tau(\varepsilon) \leq \dfrac{1}{10}\, d_{\max}^{14}\, M^9 \left( M\log(M) +
 \log(\varepsilon^{-1})\right).\]
\label{main}
\end{theorem}
This result covers many different degree sequences, for example:
\begin{itemize}
\item sparse
graphs with constant average degree and maximum degree a 
sufficiently small constant times $\sqrt{n}$, 
\item dense graphs with linear average degree and 
maximum degree a sufficiently small constant times $n$.
\end{itemize}
The mixing time bound given above is 
at most a factor of $n$ 
larger than that obtained in~\cite{CDG,CDG-corrigendum} in the regular case.
(To see this, substitute $M  = d_{\max} n$, which holds when $\boldsymbol{d}$ 
is regular: note that $M\leq d_{\max} n$ always holds, as $M$ is the sum of the
degrees.)

We expect that our approach also applies to directed graphs, which should
allow the rapid mixing proof from~\cite{directed} to be extended to
irregular directed degree sequences, under conditions analogous to those
in Theorem~\ref{main}.

\subsection{Related work}\label{s:background}

There are several approaches to the problem of sampling graphs with a given
degree sequence, though none is known to be efficient for all degree sequences.  
The configuration model of Bollob{\' a}s~\cite{bollobas}
gives expected polynomial time uniform sampling if $d_{\max} = O(\sqrt{\log n})$.
McKay and Wormald~\cite{McKW90} adapted the configuration model to give an algorithm which 
performs uniform sampling from
$\Omega(\boldsymbol{d})$ in expected polynomial time when $d_{\max} = O(M^{1/4})$.

Jerrum and Sinclair~\cite{JS90} used a construction of Tutte's to reduce the problem
of sampling from $\Omega(\boldsymbol{d})$ to the problem of sampling perfect matchings
from an auxilliary graph. The resulting Markov chain algorithm is rapidly mixing
if the degree sequence $\boldsymbol{d}$ is \emph{stable}: see~\cite{JMS}.
Stable sequences are those in which small local changes to the degree sequences
do not greatly affect the size of $|\Omega(\boldsymbol{d})|$.
Many degree sequences which satisfy the conditions of Theorem~\ref{main} are
stable;  however, not all stable sequences satisfy the conditions
of Theorem~\ref{main}. (For example,
if $d_{\min} = n/9$ and $d_{\max} =4n/9$ then $\boldsymbol{d}$
is stable~\cite{JMS} but then $\sqrt{M}\leq 2n/3$, which is
not large enough for Theorem~\ref{main}.)

Steger and Wormald~\cite{SW} gave an easily-implementable algorithm
for sampling regular graphs, and proved that their algorithm performs
asymptotically uniform sampling in polynomial time when $d=o(n^{1/28})$
(where $d$ denotes the degree).
Kim and Vu~\cite{KV} gave a sharper analysis and established that $d=o(n^{1/3})$
suffices for efficient asymptotically uniform sampling. Bayati, Kim and Saberi~\cite{BKS}
extended Steger and Wormald's algorithm to irregular degree sequences, giving
polynomial-time asymptotically uniform sampling
when $d_{\max} = o(M^{1/4})$.  From this they constructed
a sequential importance sampling algorithm for $\Omega(\boldsymbol{d})$. 
Recently, Zhao~\cite{zhao} described and analysed a similar 
approach to that of~\cite{McKW90}, in a general combinatorial setting. 
Zhao shows that for sampling from $\Omega(\boldsymbol{d})$, 
when $d_{\max} = o(M^{1/4})$, his algorithm performs 
asymptotically uniform sampling in time $O(M)$.

Finally we note that Barvinok and Hartigan~\cite{BH} showed that the 
adjacency matrix of a random element of $\Omega(\boldsymbol{d})$ is ``close'' 
to a certain ``maximum entropy matrix'', when the degree sequence is \emph{tame}.
The definition of tame depends on the maximum entropy matrix, but a sufficient
condition is that $d_{\min}\geq \alpha (n-1)$ and $d_{\max} \leq \beta (n-1)$
for some constants $\alpha,\beta > 0$.  Some degree sequences satisfying this
latter condition are stable sequences, 
and many of these degree sequences also satisfy the condition of Theorem~\ref{main}.  
It would be interesting to
explore further the connections between stable degree sequences,  tame degree 
sequences and the mixing rate of the switch Markov chain.

It is not known whether the corresponding counting problem
(exact evaluation of $|\Omega(\boldsymbol{d})|$) is $\#P$-complete.
There are several results giving asymptotic enumeration formulae
for $|\Omega(\boldsymbol{d})|$ under various conditions on $\boldsymbol{d}$:
see for example~\cite{BH,ranX,McKW91} and references therein.

\section{The switch chain and multicommodity flow}\label{s:undir}

A transition of the switch chain on $\Omega(\boldsymbol{d})$ is performed
as follows: from the current state $G\in\Omega(\boldsymbol{d})$, choose
an unordered pair of 
two distinct non-adjacent edges uniformly at random, 
say $F=\{\{ x,y\},\{ z,w\}\}$, and choose a perfect matching $F'$
from the set of three perfect matchings of (the complete graph on)
$\{ x,y,z,w\}$, chosen uniformly at random.  If 
$F'\cap \left( E(G) \setminus F\right) = \emptyset$ then the next state
is the graph $G'$ with edge set $\left(E(G) \setminus F\right)\cup F'$,
otherwise the next state is $G'=G$.

Define $M_2 = \sum_{j=1}^n d_j(d_j-1)$. 
If $P(G,G')\neq 0$ and $G\neq G'$ then $P(G,G') = \frac{1}{3 a(\boldsymbol{d})}$, where
\begin{equation}
\label{ad}
 a(\boldsymbol{d}) = \binom{M/2}{2} - \dfrac{1}{2}\, M_2 
\end{equation}
is the number of unordered pairs of distinct nonadjacent edges in $G$.
This shows that the Markov chain is symmetric.
The chain is aperiodic since by definition $P(G,G)\geq 1/3$ for all
$G\in\Omega(\boldsymbol{d})$.

\subsection{Multicommodity flow}\label{ss:flow}

To bound the mixing time of the switch chain, we apply a multicommodity
flow argument.  Suppose that $\mathcal{G}$ is the graph underlying a Markov
chain $\mathcal{M}$, so that $xy$ is an edge of $\mathcal{G}$ if and only
if $P(x,y)>0$. A \emph{flow} in $\mathcal{G}$ is a function $f:\mathcal{P}\rightarrow [0,\infty)$ such that
\[ \sum_{p\in\mathcal{P}_{xy}} f(p) = \pi(x)\pi(y) \quad \text{ for all }
 \,\, x,y\in\Omega,\,\, x\neq y.\]
Here $\mathcal{P}_{xy}$ is the set of all simple directed paths from $x$ to
$y$ in $\mathcal{G}$ and $\mathcal{P} = \cup_{x\neq y} \mathcal{P}_{xy}$.
Extend $f$ to a function on oriented edges by setting
$f(e) = \sum_{p\ni e} f(p)$,
so that $f(e)$ is the total flow routed through $e$.  Write $Q(e) = \pi(x) P(x,y)$
for the edge $e=xy$.  Let $\ell(f)$ be the length of the longest path with
$f(p) > 0$, and let $\rho(e) = f(e)/Q(e)$ be the \emph{load} of the edge $e$.
The \emph{maximum load} of the flow is
$\rho(f) = \max_e \rho(e)$.
Using Sinclair~\cite[Proposition 1 and Corollary 6']{sinclair},  the mixing time of
$\mathcal{M}$ can be bounded above by
\begin{equation}
\label{flowbound} 
 \tau(\varepsilon) \leq \rho(f)\ell(f)\left(\log(1/\pi^*) + \log(\varepsilon^{-1}\right)\end{equation}
where $\pi^* = \min\{ \pi(x) \mid x\in\Omega\}$.

\subsection{Defining the flow}

The definition of the multicommodity flow given in~\cite[Section 2.1]{CDG} carries across
to irregular degree sequences without change.  This is because the flow from $G$ to $G'$
depends only on the symmetric difference $G\triangle G'$ of $G$ and $G'$, 
treated as a 2-edge-coloured
graph (with edges from $G\setminus G'$ coloured blue and edges from 
$G'\setminus G$ coloured red, say). 
The blue degree at a given vertex equals the red degree at that vertex, but 
in general the blue degree sequence will not be regular.
Hence the multicommodity flow definition given in~\cite{CDG}
is already general enough to handle irregular degree sequences.

The multicommodity flow is defined using a process which we now sketch.
Given $G, G'\in\Omega(\boldsymbol{d})$:
\begin{itemize}
\item Define a bijection from the set of blue edges incident at $v$ to the
set of red edges incident at $v$, for each vertex $v\in [n]$.
The vector of these bijections is called a \emph{pairing} $\psi$,
and the set of all possible pairings is denoted $\Psi(G,G')$.
\item  The pairing gives a canonical way to decompose the symmetric difference
$G\triangle G'$ into a sequence of \emph{circuits}, where each circuit is a 
blue/red-alternating closed walk.
\item Each circuit is decomposed in a canonical way
into a sequence of simpler circuits
of two types: 1-\emph{circuits} and 2-\emph{circuits}.  A 1-circuit is
an alternating cycle in $G\triangle G'$, while a 2-circuit is
an alternating walk with one vertex of degree 4, the rest of degree 2,
consisting of two odd cycles which share a common vertex.
Each 1-circuit or 2-circuit has a designated \emph{start vertex}.
(The start vertex of a 2-circuit is the unique vertex of degree 4.)
An important fact is that the 1-circuits and 2-circuits are
\emph{pairwise edge-disjoint}.
\item  Each 1-circuit or 2-circuit is processed in a canonical way to
give a segment of the canonical path $\gamma_\psi(G,G')$.
\end{itemize}
For full details see~\cite[Section 2.1]{CDG}.

\section{Analysing the flow}\label{s:analysis}

Now we show how to bound the load of the flow by adapting the analysis 
from~\cite{CDG}. 
Note that some proofs in~\cite{CDG} used the assumption $d=d(n)\leq n/2$,
since (for regular sequences) the general result follows by complementation.  
This trick does
not work for irregular degree sequences, so we cannot make a similar
assumption here.

Given matrices $G$, $G'$, $Z\in\Omega(\boldsymbol{d})$, define the
\emph{encoding} $L$ of $Z$ (with respect to $G,G'$) by
\[ L + Z = G + G'\]
by identifying each of $Z$, $G$ and $G'$ with their symmetric 0-1 adjacency matrices.
Then $L$ is a symmetric $n\times n$ matrix with entries in $\{ -1, 0, 1,2\}$
and with zero diagonal.
Entries which equal $-1$ or 2 are called \emph{defect entries}.  Treating $L$ as an
edge-labelled graph with edges labelled $-1, 1, 2$ (and omitting edges corresponding to
zero entries), a \emph{defect edge} is an edge labelled $-1$ or $2$.
(In~\cite{CDG} these were called ``bad edges''.)
Specifically, we will refer to $(-1)$-\emph{defect edges} and to 
$2$-\emph{defect edges}.  A 2-defect edge is present in both $G$ and
$G'$ but is absent in $Z$, while a $(-1)$-defect edge is absent in
both $G$ and $G'$ but is present in $Z$.

We say that the \emph{degree} of vertex $v$ in $L$ is the sum of the
labels of the edges incident with $v$ (equivalently, the sum of the
entries in the row of $L$ corresponding to $v$).  By definition,
the degree sequence of $L$ equals $\boldsymbol{d}$.

Some proofs from~\cite{CDG,CDG-corrigendum} also apply in the
irregular case without any substantial change (after replacing
$d$ by $d_{\max}$). These proofs refer only to the symmetric difference 
and the process used to construct the multicommodity
flow (and none of them use the assumption $d\leq n/2$).  
 We state two of these results now. 


\begin{lemma}
Suppose that $G,G',Z,Z'\in\Omega(\boldsymbol{d})$ are such that $(Z,Z')$ is
a transition of the switch chain which lies on the canonical path 
$\gamma_\psi(G,G')$ for some $\psi\in\Psi(G,G')$.
Let $L$ be the encoding of $Z$ with respect to $(G,G')$.  Then the following statements
hold:
\begin{enumerate}
\item[\emph{(i)}]
\emph{(\cite[Lemma 1]{CDG})}\ From $(Z,Z')$, $L$ and $\psi$ it is possible to uniquely recover $G$ and $G'$.
\label{unique}
\item[\emph{(ii)}] 
\emph{(\cite[Lemma 2]{CDG})}\
There are at most four defect edges in $L$.  The labelled
graph consisting of the defect edges in $L$ must form a subgraph of one
of the five labelled graphs shown in Figure~\ref{f:possible},
where \emph{``?''} represents a label which may be either $-1$ or 2.
\label{defects}
\end{enumerate}
\label{oldstuff}
\end{lemma}



\begin{figure}
\begin{center}
\psfrag{a}{2}\psfrag{x}{$-1$}\psfrag{z}{?}
\includegraphics[scale=0.45]{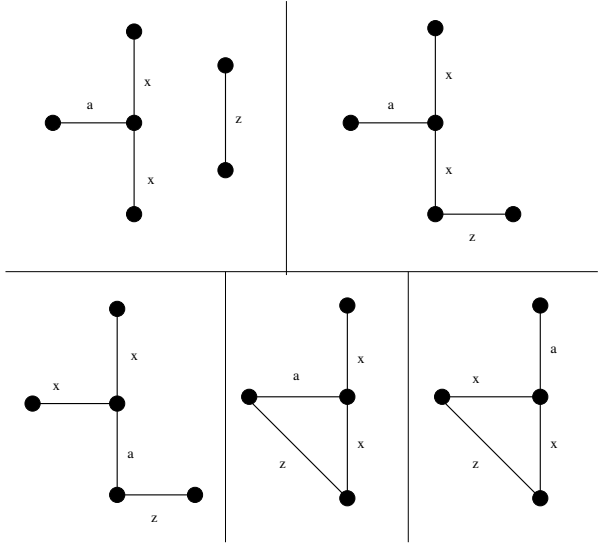}
\caption{The five possible configurations of four defect edges}
\label{f:possible}
\end{center}
\end{figure}


The next result collects together some further useful results about
encodings.  

\begin{lemma}
Suppose that the conditions of Lemma~\ref{oldstuff} hold.
Let $x,y,z\in [n]$ be distinct vertices.  
\begin{itemize}
\item[\emph{(i)}] If $L(x,y)=2$ then $d_x\geq 2$, $d_y\geq 2$.  
\item[\emph{(ii)}] If $L(x,y)=2$ and $L(y,z)=2$ then $d_y\geq 4$.
\item[\emph{(iii)}]  If $L(x,y)=2$ and $L(y,z)=-1$ then $d_y\geq 3$.
\end{itemize}
\label{structure}
\end{lemma}
%
\begin{proof}
It follows from the definition of the multicommodity flow given in~\cite{CDG}
that a 2-defect edge $\{ x,y\}$ (with $L(x,y)=2$) can only arise in
two cases: 
\begin{itemize}
\item[(a)]
$\{x,y\}$ is a \emph{shortcut edge} 
which is present in $G, G'$ but which is
absent in $Z$.  (See~\cite[Figure 4]{CDG}.)
In this case, $x$ and $y$ are vertices on some
2-circuit, which is an alternating blue/red walk in the symmetric difference
$G\triangle G'$.  Hence both $x$ and $y$ have degree at least two
in $G$.
\item[(b)] $\{ x,y\}$ is an \emph{odd chord} which is present in $G, G'$
but which is absent in $Z$.    
(See the section ``Processing a 1-circult'' in~\cite{CDG}.) 
In this case, $x$ and $y$ are vertices
on some 1-circuit, which is an alternating blue/red walk in the
symmetric difference $G\triangle G'$.  
Again, this shows that both $x$ and $y$ have degree at least two in $G$.
\end{itemize}
This proves (i).  

Next, if $y$ is incident with two edges
of defect 2 then it must be that 
one is an odd chord for a 1-circuit $C_1$ and
one is a shortcut edge for a 2-circuit $C_2$.
Then $y$ is incident in $G$ with an edge of $C_1$, an edge of $C_2$
and the two edges $\{ x,y\}$, $\{y,z\}$ which are 2-defect edges in $L$.
Since $C_1$ and $C_2$ are edge-disjoint and no defect edge belongs
to $G\triangle G'$, it follows that $d_y\geq 4$, proving (ii).

We may adapt this argument to prove (iii), 
noting that a $(-1)$-defect may only
arise from a shortcut edge or an odd chord which is
absent in $G$ and $G'$ and present in $Z$.   
\hfill \qedbox
\end{proof}

Now we extend the term ``encoding'' to refer to any symmetric
$n\times n$ matrix with entries in $\{ -1, 0, 1,2\}$ which
has zero diagonal and row sums given by $\boldsymbol{d}$.
We say that an encoding $L$ is \emph{consistent} with $Z$ 
if $L+Z$ only takes entries in $\{0,1,2\}$.
Say that an encoding is \emph{valid} if it satisfies the conclusions
of Lemma~\ref{oldstuff}(ii), and that a valid encoding is \emph{good}
if it also satisfies the conclusion of Lemma~\ref{structure}. 
Let $\mathcal{L}(Z)$ be the set of valid encodings which are consistent with
$Z$, and let $\mathcal{L}^*(Z)$ be the set of good encodings which are
consistent with $Z$.  In~\cite{CDG} the set $\mathcal{L}(Z)$ was
studied, but we require a bit more information about our encodings,
so we will focus on the smaller set $\mathcal{L}^*(Z)$.

\begin{lemma}
\emph{(\cite[Lemma 5]{CDG} and~\cite[Lemma 1]{CDG-corrigendum})}\
The load $f(e)$ on the transition $e=(Z,Z')$ satisfies
\[ f(e)\leq  d_{\max}^{14}\,{\frac{|\mathcal{L}^*(Z)|}{|\Omega(\boldsymbol{d})|}}.\]
\label{fbound}
\end{lemma}
%
\begin{proof}
In~\cite[Lemma 5]{CDG} and~\cite[Lemma 1]{CDG-corrigendum}
it was shown that $f(e)\leq d^{14}\, 
  |\mathcal{L}(Z)|/|\Omega(\boldsymbol{d})|^2$ when 
$\boldsymbol{d} = (d,d,\ldots, d)$ is a regular degree sequence.
(The assumption $d\leq n/2$ is not used in this proof.)  
The proof relied on the fact that $\mathcal{L}(Z)$ contains all encodings
which may arise along a canonical path.  But the same is true for
$\mathcal{L}^\ast(Z)$, by Lemma~\ref{oldstuff}(ii) and Lemma~\ref{structure}, 
so the proof goes through without change in
the irregular setting (after replacing $d$ by $d_{\max}$).  \hfill \qedbox
\end{proof}

The switch operation can be extended to encodings 
in the natural way: each switch reduces two
edge labels by one and increases two edge labels by one, without changing the degrees.
It was shown in~\cite[Lemma 3]{CDG} that from any valid encoding,
one could obtain a graph (with no defect edges)
by applying a sequence of at most three switches.  
In~\cite[Lemma 4]{CDG} we used this fact to bound the ratio
$|\mathcal{L}(Z)|/|\Omega(\boldsymbol{d})|$ for regular degree sequences.  
This provided an upper bound for the flow $f(e)$ through
a transition $e=(Z,Z')$ (as in Lemma~\ref{fbound}, above).

The proof of~\cite[Lemma 3]{CDG} uses regularity to prove 
the existence of certain edges which are
needed in order to find switches to remove the defect edges.
This argument fails for irregular degree sequences.  
Instead, we consider
a slightly more complicated operation than a switch, which we call a 
\emph{3-switch}. (This operation is called a
``circular $C_6$-swap'' in~\cite{EKMS}).  

A 3-switch is described by a 6-tuple $(a_1,b_1,a_2,b_2,a_3,b_3)$ of distinct vertices
such that $a_1b_1$, $a_2b_2$, $a_3b_3$ are all edges and $a_2b_1$, $a_3b_2$, $a_1b_3$ are
all non-edges.   The 3-switch deletes the three edges $a_1b_1$, $a_2b_2$, $a_3b_3$ from the
edge set and replaces them with $a_2b_1$, $a_3b_2$, $a_1b_3$, as shown in Figure~\ref{3-switch}.

\begin{figure}[ht]
\begin{center}
\unitlength=1cm
\begin{picture}(10,3)(0,0)
\put(0.5,0.5){
\includegraphics[scale=0.7]{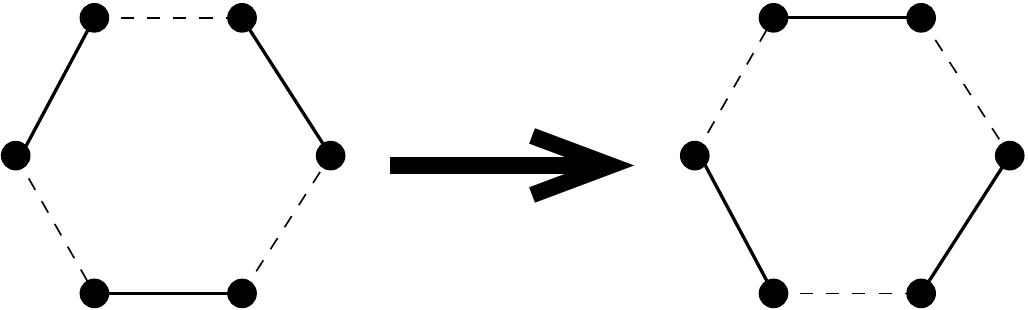}
}
\put(0.2,1.5){$a_1$}
\put(0.9,2.8){$b_1$}
\put(2.5,2.8){$a_2$}
\put(3.0,1.8){$b_2$}
\put(2.3,0.2){$a_3$}
\put(1.1,0.2){$b_3$}
\put(5.1,1.7){$a_1$}
\put(5.7,2.8){$b_1$}
\put(7.3,2.8){$a_2$}
\put(8.1,1.5){$b_2$}
\put(7.2,0.2){$a_3$}
\put(5.9,0.2){$b_3$}
\end{picture}
\caption{A 3-switch}
\label{3-switch}
\end{center}
\end{figure}
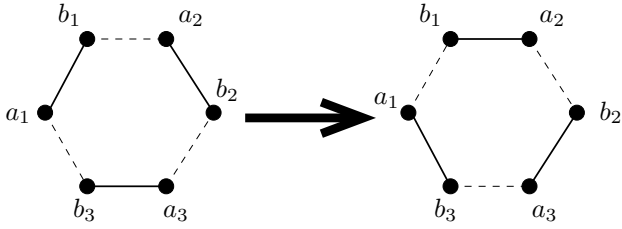
Let $\mathcal{C}(p,q)$ be the set of encodings in $\mathcal{L}^\ast(Z)$ with 
precisely $p$ defect edges labelled 2 and precisely $q$ defect edges
labelled $-1$, for $p\in\{0,1,2\}$ and $q\in\{0,1,2,3\}$. 
Then $\Omega(\boldsymbol{d})= \mathcal{C}(0,0)$ and 
\[ \mathcal{L}^\ast(Z) = \cup_{p=0}^2\,\cup_{q=0}^3\, \mathcal{C}(p,q),\]
where this union is disjoint.
(Note that $\mathcal{C}(2,3)=\emptyset$, by Lemma~\ref{oldstuff}(ii).)

For $v\in [n]$, given an encoding $L$,  write $N_L(v)$ to denote
the set of $w\in [n]\setminus \{v\}$ such that $L(v,w)=1$.
This is the set of neighbours of $v$ in $L$, where neighbours along
defect edges are not included. If $L\in \mathcal{C}(p,q)$ then there
are precisely $M/2 - 2p + q$ non-defect edges in $L$.  (To see this,
note that the sum of all entries in the matrix $L$ must equal $M$,
and $L$ has zero diagonal.)

\begin{lemma}
Suppose that $\boldsymbol{d}$ satisfies $d_{\min}\geq 1$ and $3\leq d_{\max}\leq \nfrac{1}{4}\sqrt{M}$.
Let $Z\in\Omega(\boldsymbol{d})$.  Then
\[ |\mathcal{L}^\ast(Z)| \leq  \dfrac{1}{5}\, M^6\,  |\Omega(\boldsymbol{d})|.\]
\label{3switches}
\end{lemma}

\begin{proof}
We prove that any $L\in\mathcal{L}^\ast(Z)$ can be transformed into
an element of $\Omega(\boldsymbol{d})$ (with no defect edges) using a sequence
of at most three 3-switches.  The strategy is as follows:  in Phase 1 we
aim to remove two defects per 3-switch (one 2-defect and one $(-1)$-defect),
then in Phase 2 we remove one 2-defect per 3-switch, and finally in Phase 3
we remove one $(-1)$-defect per 3-switch.
There is at most one step in Phase 1, though the other phases may have
more than one step: any phase may be empty.
Each 3-switch we perform gives rise to an upper bound on certain ratios
of the sizes of the sets $\mathcal{C}(p,q)$, by double counting.  
The proof is completed by combining these bounds.
(Such an argument is often called a ``switching argument'' in the
asymptotic enumeration literature: see~\cite{McKW91} for example.)

\medskip

\noindent {\bf Phase 1.}\ 
If $p+q\leq 3$ then Phase 1 is empty: proceed to Phase 2.
Otherwise, suppose that $L\in\mathcal{C}(p,q)$ where
$p+q =4$.  
Then $(p,q)\in \{ (2,2), \, (1,3)\}$, 
and it follows from Figure~\ref{f:possible} that there must be 
a vertex $b_1$ which is incident with a 2-defect $L(a_1,b_1) = 2$
and a $(-1)$-defect $L(a_2,b_1) = -1$.  
We count the number of 3-switches $(a_1,b_1,a_2,b_2,a_3,b_3)$ which
may be applied to $L$ to produce an encoding $L'\in\mathcal{C}_{p-1,q-1}$.
This operation is shown in Figure~\ref{double-switch}, where defect
edges are shown using thicker lines: a thick solid line is a 
2-defect edge while a thick dashed line is a $(-1)$-defect edge.

\begin{figure}[ht]
\begin{center}
\unitlength=10mm
\begin{picture}(10,3)(0,0)
\put(0.5,0.5){
\includegraphics[scale=0.7]{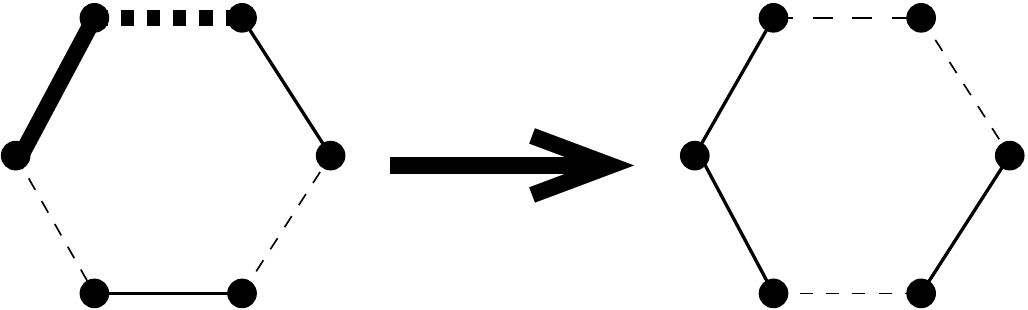}
}
\put(0.2,1.5){$a_1$}
\put(0.9,2.8){$b_1$}
\put(2.5,2.8){$a_2$}
\put(3.0,1.8){$b_2$}
\put(2.3,0.2){$a_3$}
\put(1.1,0.2){$b_3$}
\put(5.1,1.7){$a_1$}
\put(5.7,2.8){$b_1$}
\put(7.3,2.8){$a_2$}
\put(8.1,1.5){$b_2$}
\put(7.2,0.2){$a_3$}
\put(5.9,0.2){$b_3$}
\end{picture}
\caption{A 3-switch with $L(a_1,b_1)=2$,\,  $L(a_2,b_1)=-1$.}
\label{double-switch}
\end{center}
\end{figure}
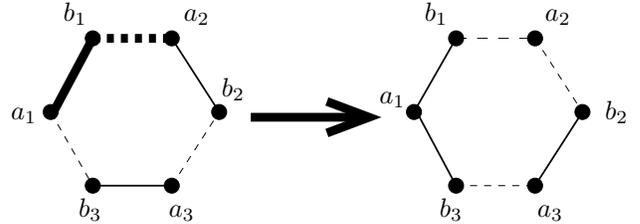
Given $(a_1,b_1,a_2)$, there
is at least one vertex $b_2\in N_L(a_2)\setminus \{ a_1\}$.
To see this, first suppose that $a_2$ is not incident with a 2-defect.
Then $N_L(a_2)$ has at least $d_{a_2}+1\geq 2$ elements, leaving at least
one which is distinct from $a_1$.  Otherwise, if $a_2$ is incident
with a 2-defect then it can be incident with at most one 2-defect,
since $p\leq 2$.  Then there are at least $d_{a_2}-2$ choices for $b_2$
in $N_L(a_2)\setminus \{ a_1\}$,
and this number is positive by Lemma~\ref{structure}(iii).

Next, we choose $(a_3,b_3)$ such that all six vertices are distinct,
$L(a_3,b_3)=1$ and  $L(a_3,b_2)=L(a_1,b_3)=0$.
There are $M - 4p + 2q$ possibilities for $(a_3,b_3)$ with $L(a_3,b_3)=1$,
but we must reject those which are 
incident with the four vertices already chosen, or which are
incident to a neighbour of $a_1$ or $b_2$.  We need to be careful
with $(-1)$-defect edges.  Hence, for all $x\in [n]$, let $\eta_x$
be the number of $(-1)$-defect edges \emph{other than} $\{ a_2,b_1\}$
which are incident with $x$ in $L$.  Then $\sum_{x\in [n]} \eta_x \leq 4$
since there are at most two more $(-1)$-defect edges in $L$.
Furthermore, $\eta_{a_1} + \eta_{b_2}\leq 3$.
The number of bad choices for $(a_3,b_3)$ is at most
\[ 2\left( |N_L(b_1)| + \sum_{x\in N_L(a_1)} |N_L(x)| + \sum_{y\in N_L(b_2)}
       | N_L(y)|\right).
\]
To see this, note that $a_2\in N_L(b_2)$ so all edges incident with $a_2$
are counted in the final sum (with $y=a_2$).  Furthermore, for each $x\in N_L(a_1)$,
the edge from $a_1$ to $x$ is among those counted by $|N_L(x)|$, so the first
sum covers all edges incident with $a_1$ or a neighbour of $a_1$
(and similarly for the second sum).  Hence the number of bad choices
for $(a_3,b_3)$ is at most
\begin{align*}
 & 2\biggl( d_{b_1} - 1 + \eta_{b_1} 
    + \sum_{x\in N_L(a_1)} ( d_x  + \eta_x) \\
 & \hspace*{45mm} {} + 
       \sum_{y\in N_L(b_2)} (d_y  + \eta_y)\biggr)\\
  &\leq 2\biggl( d_{\max} -1 + \eta_{a_2} + \eta_{b_1} \\
  & \hspace*{15mm} {}
        + d_{\max} \left(d_{a_1} + d_{b_2} - 2 + \eta_{a_1} + \eta_{b_2}\right)\\
  & \hspace*{45mm} {}
    + \sum_{x\not\in \{ a_1,b_1,a_2,b_2\}} 2\eta_x \biggr)\\
 &\leq 2\left(2 d_{\max}^2 + 2d_{\max} + 1\right).
\end{align*}
The final inequality follows from setting $\eta_{a_1} + \eta_{b_2} = 3$,
the maximum possible, and letting $\eta_x=1$ for some $x\not\in \{ a_1,b_1,a_2,b_2\}$
(as well as  bounding $d_{a_1}$ and $d_{b_2}$ by $d_{\max}$).

Hence, the number of possible 3-switches $(a_1,b_1,a_2,b_2,a_3,b_3)$
such that $L(a_1,b_1)=2$ and
$L(a_1,b_3)=-1$ is at least 
\begin{align}
\label{21forward}
  M -4p & +2q - 2\left(2 d_{\max}^2 + 2 d_{\max} + 1\right) \\
 &\geq  M - 2\left(2 d_{\max}^2 + 2 d_{\max} + 3\right) \notag\\
 &\geq M-6 d_{\max}^2 \notag\\ 
  &\geq M/2 \notag
\end{align}
since $3\leq d_{\max}\leq \nfrac{1}{4}\, \sqrt{M}$.
Each such 3-switch produces an encoding $L'\in\mathcal{C}(p-1,q-1)$.

Now we consider the reverse of this operation, which is
given by reversing the arrow in Figure~\ref{double-switch}.  Given 
$L'\in\mathcal{C}(p-1,q-1)$, we need an upper bound on
the number of 6-tuples
$(a_1,b_1,a_2,b_2,a_3,b_3)$ such that $L'(a_1,b_1)=L'(a_1,b_3)=L'(a_3,b_2)=1$
and $L'(a_2,b_1)=L'(a_2,b_2)=L'(a_3,b_3)=0$.   Since the encoding
$L\in\mathcal{C}(p,q)$ produced by this reverse operation must be
consistent with $Z$, it follows that $\{ a_2,b_1\}$ must be an
edge of $Z$.  Hence there are precisely $M$ choices for $(a_2,b_1)$.
There are at most $d_{b_1} + \eta_{b_1}$ ways to choose $a_1\in N_L(b_1)$
and at most $d_{a_1}-1+\eta_{a_1}$ ways to choose $b_3\in N_L(a_1)\setminus
\{ b_1\}$.  
From Figure~\ref{f:possible}, if $\eta_{a_1}=2$ then $\eta_{b_1}=0$,
and if $\eta_{b_1}=1$ then $\eta_{a_1}\leq 1$.  
Furthermore, $\eta_{b_1}\leq 1$.  (Otherwise, the reverse switching
would produce an encoding which is not valid.)
Therefore,
\begin{align*}
   (d_{b_1} + \eta_{b_1})(d_{a_1} - 1 + \eta_{a_1})
    &\leq d_{\max}\, (d_{\max} + 1)\\
  & \leq \nfrac{4}{3}\, d_{\max}^2.
\end{align*}

Finally we
must choose $(a_3,b_2)$ such that $L(a_3,b_2) = 1$, the vertices $a_3,b_2$
are distinct from the four vertices chosen so far and $L'(a_2,b_2) = L'(a_3,b_3)=0$.  
 When $(p,q)=(2,2)$ we ignore all conditions
except $L(a_3,b_2)=1$, and take
\[ M - 4(p-1) + 2(q-1) = M - 2 \leq M\]
as an upper bound for the number of good choices of $(a_3,b_2)$.
When $(p,q)=(1,3)$ there are no 2-defects in $L'$, as
$L'\in\mathcal{C}(0,2)$, so there are at most
\begin{align*}  M - 4(p-1) & + 2(q-1) - (d_{a_1} + d_{b_1} + d_{a_2} + d_{b_3})\\
 & \leq M -4p + 2q - 2 \\
  & = M
\end{align*}
good choices for $(a_3,b_2)$.  (The existence of any additional $(-1)$-defect
edges incident with $a_1$, $b_1$, $a_2$ or $b_3$ can only help here.)
Hence the number of ways to apply the reverse operation to $L'\in\mathcal{C}(p-1,q-1)$ to produce a consistent encoding $L\in\mathcal{C}(p,q)$ is at
most $\nfrac{4}{3}\, d_{\max}^2 M^2$.

Combining this with (\ref{21forward}) shows that whenever $p+q=4$,
by double counting,
\begin{equation}
\label{B21}
 \frac{|\mathcal{C}(p,q)|}{|\mathcal{C}(p-1,q-1)|} \leq \dfrac{8}{3}\, d_{\max}^2 M.
\end{equation}

\medskip

\noindent {\bf Phase 2.}\  Once Phase 1 is complete, we have reached an
encoding $L\in\mathcal{C}(p,q)$ with $p+q\leq 3$.
If $p=0$ then Phase 2 is empty: proceed to Phase 3.
Otherwise, we have
$(p,q) \in \{  (2,1),\, (2,0),\, (1,2),\, (1,1),\, (1,0)\}$.
We count the number of ways to perform a 3-switch to reduce the
number of 2-defect edges by one, as shown in Figure~\ref{2-3-switch}.

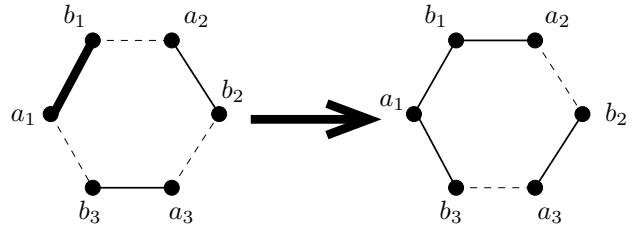
\begin{figure}[ht]
\begin{center}
\unitlength=1cm
\begin{picture}(10,3)(0,0)
\put(0.5,0.5){
\includegraphics[scale=0.7]{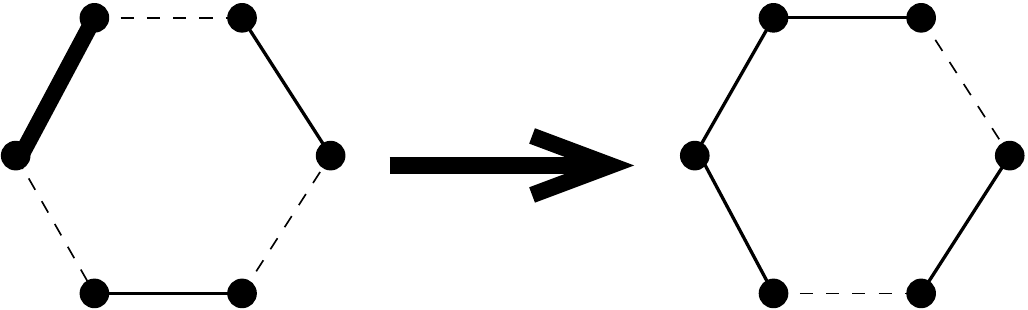}
}
\put(0.2,1.5){$a_1$}
\put(0.9,2.8){$b_1$}
\put(2.5,2.8){$a_2$}
\put(3.0,1.8){$b_2$}
\put(2.3,0.2){$a_3$}
\put(1.1,0.2){$b_3$}
\put(5.1,1.7){$a_1$}
\put(5.7,2.8){$b_1$}
\put(7.3,2.8){$a_2$}
\put(8.1,1.5){$b_2$}
\put(7.2,0.2){$a_3$}
\put(5.9,0.2){$b_3$}
\end{picture}
\caption{A 3-switch with $L(a_1,b_1)=2$. }
\label{2-3-switch}
\end{center}
\end{figure}

Choose an ordered pair $(a_1,b_1)$ such that $L(a_1,b_1)=2$, in 
$2p$ ways. The number of ways to choose the ordered pair $(a_2,b_2)$ such that 
$a_1,b_1,a_2,b_2$ are all distinct, $L(a_2,b_2)=1$ and $L(a_2,b_1)=0$,
is  at least
\begin{align*}
&M - 4p + 2q - 2\biggl( |N_L(a_1)| + \sum_{x\in N_L(b_1)} |N_L(x)|\biggr)\\
  &\geq M - 4p + 2q \\
   & \hspace*{1cm} {} - 2\biggl( d_{a_1} - 2 + \eta_{a_2}
   + \sum_{x\in N_L(b_1)} (d_{x} + \eta_x)\biggr)\\
 &\geq M - 2\biggl(d_{\max} + 2 
     + d_{\max} (d_{b_1} - 2 + \eta_{b_1}) 
    + \sum_{x\neq b_1} \eta_x \biggr)\\
 &\geq M - 2\left( d_{\max}^2 + d_{\max} + 4\right)\\
 & \geq M - 4 d_{\max}^2.
\end{align*}
This uses the fact that $L$ may contain up to two $(-1)$-defect
edges, so the worst case is when $\eta_{b_1} = 2$ and 
$\sum_{x\neq b_1} \eta_x = 2$.

Next, choose an ordered pair $(a_3,b_3)$ such that all six vertices
are distinct, $L(a_3,b_3)=1$ and $L(a_1,b_3)=L(a_3,b_2)=0$.  
This can be done in at least
\begin{align*}
 &M - 4p + 2q \\
  & \hspace*{3mm} {} - 2\biggl( |N_L(b_1)| + \sum_{x\in N_L(a_1)} |N_L(x)|
         + \sum_{y\in N_L(b_2)} | N_L(y)|\biggr)\\
  &\geq M - 4p + 2q 
   - 2\biggl( d_{b_1} -2 + \eta_{b_1} + 
   \sum_{x\in N_L(a_1)} (d_x  + \eta_x) 
  \\ & \hspace*{4cm} {} 
  + \sum_{y\in N_L(b_2)} (d_y + \eta_y)\bigg)\\
 &\geq M - 2\biggl(d_{\max} + 2 
    + \eta_{a_2} + \eta_{b_1} 
   \\ & \hspace*{25mm} {} 
  + d_{\max}\left( d_{a_1} + d_{b_2} - 2 + \eta_{a_1} + \eta_{b_2}\right)  \\
  & \hspace*{40mm} {} 
  +  \sum_{x\not\in \{ a_1,b_1,a_2,b_2\}} 2\eta_x \biggr)\\
 &\geq M - 2\left(2 d_{\max}^2 + 2d_{\max} + 4\right)\\
  & \geq M - 8 d_{\max}^2
\end{align*}
ways, arguing as above.  
(Again, the worst case is when $\eta_{a_1} + \eta_{b_2} = 3$ and $\eta_x=1$
for some $x\not\in \{ a_1,b_1,a_2,b_2\}$.)
Hence there are at least
\begin{equation}
\label{forward2}
  2\left(M - 4 d_{\max}^2\right)\left(M - 8d_{\max}^2\right)
  \geq \nfrac{1}{2} M^2 
\end{equation}
such choices for $(a_1,b_1,a_2,b_2,a_3,b_3)$, using the stated
upper bound on $d_{\max}$.

For the reverse operation, let $L'\in\mathcal{C}(p-1,q)$ where
$(p,q)\in \{ (2,1),\, (2,0),\, (1,2),\, (1,1),\, (1,0)\}$. 
 We need an upper bound
on the number of 6-tuples $(a_1,b_1,a_2,b_2,a_3,b_3)$ with
$L(a_1,b_1)=L(a_1,b_3) = L(a_2,b_1) = L(a_3,b_2)=1$ and
$L(a_2,b_2)=L(a_3,b_3)=0$.  There are at most $M-4p+2q\leq M$
choices for $(a_1,b_1)$ with $L(a_1,b_1)=1$, and then there are
at most 
\[ (d_{a_1} - 1 + \eta_{a_1})(d_{b_1} - 1 + \eta_{b_1}) \leq d_{\max}^2\]
choices for $(a_2,b_3)$.     This uses the fact that there
are at most two defect edges in $L'$, and hence 
$\eta_{a_1} + \eta_{b_1} \leq 2$, by choice of $(a_1,b_1)$.
Finally there are
at most $M - 4p + 2q\leq M$ choices for $(a_3,b_2)$,  so the number of
6-tuples where the reverse operation can be performed is at most
$d_{\max}^2 M^2$.

Combining this with (\ref{forward2}), it follows that 
\begin{equation} 
\frac{|\mathcal{C}(p,q)|}{|\mathcal{C}(p-1,q)|} \leq 2 d_{\max}^2
\label{B2}
\end{equation}
holds for $(p,q)\in \{ (2,1),\, (2,0),\, (1,2),\, (1,1),\, (1,0)\}$.

\medskip

\noindent {\bf Phase 3.}\ After Phase 2, we may
suppose that $p=0$.  Let $L\in\mathcal{C}(0,q)$ where $q\in \{ 1,2,3\}$.  
We count the number of 6-tuples $(a_1,b_1,a_2,b_2,a_3,b_3)$
where a 3-switch can be performed with $L(a_2,b_1)=-1$.
Performing this 3-switch
will produce $L'\in\mathcal{C}(0,q-1)$, as illustrated in Figure~\ref{1switch}.

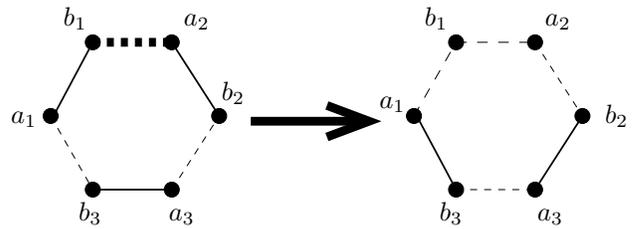
\begin{figure}[ht]
\begin{center}
\unitlength=1cm
\begin{picture}(10,3)(0,0)
\put(0.5,0.5){
\includegraphics[scale=0.7]{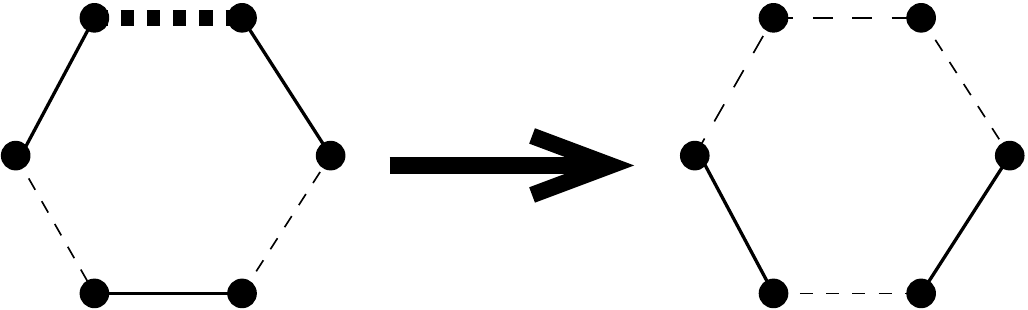}
}
\put(0.2,1.5){$a_1$}
\put(0.9,2.8){$b_1$}
\put(2.5,2.8){$a_2$}
\put(3.0,1.8){$b_2$}
\put(2.3,0.2){$a_3$}
\put(1.1,0.2){$b_3$}
\put(5.1,1.7){$a_1$}
\put(5.7,2.8){$b_1$}
\put(7.3,2.8){$a_2$}
\put(8.1,1.5){$b_2$}
\put(7.2,0.2){$a_3$}
\put(5.9,0.2){$b_3$}
\end{picture}
\caption{A 3-switch with $L(a_2,b_1)=-1$. }
\label{1switch}
\end{center}
\end{figure}
There are $2q$ ways to choose $(b_1,a_2)$, and
at least $d_{b_1}+1$ ways to choose $a_1\in N_L(b_1)$.
Then there are 
at least $d_{a_2}$ ways to choose $b_2\in N_L(a_2)\setminus \{ a_1\}$.
(Note that the presence of other $(-1)$-defect edges incident with
$b_1$ or $a_2$ only helps here.)  Finally, we must choose $(a_3,b_3)$
with $L(a_3,b_3)=1$
such that all vertices are distinct, $L(a_3,b_2)=0$ and $L(a_1,b_3)=0$.
The number of choices for $(a_3,b_3)$ is at least
\begin{align*}
 & M + 2q - 2\biggl( \sum_{x\in N_L(a_1)} |N_L(x)| + \sum_{y\in N_L(b_2)} |N_L(y)|
            \biggr)\\
  &\geq  M  + 2q - 2\biggl( \sum_{x\in N_L(a_1)} (d_x + \eta_x) 
         + \sum_{y\in N_L(b_2)} (d_y + \eta_y)\biggr)\\
 &\geq M  - 2\biggl( d_{\max} \left( 2d_{\max} + \eta_{a_1} + \eta_{b_2}\right)
       - 1 + \sum_{x\not\in \{ a_1,b_2\}} 2\eta_x\biggr)\\
 &\geq M - 2\left(2d_{\max}^2 + 3 d_{\max} + 1\right)\\
 & \geq M - 8 d_{\max}^2.
\end{align*}
The penultimate line follows by substituting $\eta_{a_1} + \eta_{b_2} =3$
and letting $\eta_x = 1$ for some $x\not\in \{ a_1,b_1,a_2,b_2\}$. 
Hence the number of 3-switches which can be performed in $L$
to reduce the number of 2-defects by exactly one
is at least
\begin{align}
\label{1forward} 
2q(d_{b_1}+1)\, d_{a_2}\, (M - 8 d_{\max}^2) &\geq 4q(M - 8 d_{\max}^2) 
\notag \\ & \geq 2M,
\end{align}
using the given bounds on $d_{\max}$.

For the reverse operation, let $L'\in \mathcal{C}(0,q-1)$, where
$q\in \{ 1,2,3\}$.  We need an upper bound on the number of
6-tuples such that $L(a_1,b_3)=L(a_3,b_2)=1$, $L(a_1,b_1)=L(b_1,a_2)=
L(a_2,b_2)=L(a_3,b_3) = 0$ and $\{ a_2,b_1\}$ is an edge of $Z$.
There are at most $M$ choices for $(a_2,b_1)$ satisfying
the latter condition, then at most
$M + 2(q-1) - 2(d_{a_2} + d_{b_1})\leq M$ ways to choose 
$(a_3,b_2)$ with $L(a_3,b_2)=1$ 
and $a_1,a_3,b_2,b_3$ all distinct. Similarly, there at most
$M$ ways to choose $(a_1,b_3)$.
Hence the number of reverse operations is at most $M^3$. 

Combining this with (\ref{1forward}) shows that
\begin{equation}
 \frac{|\mathcal{C}(0,q)|}{|\mathcal{C}(0,q-1)|} \leq \nfrac{1}{2} M^2
\label{B1}
\end{equation}
holds for $q\in \{1,2,3\}$, by double counting.

\medskip

\noindent {\bf Consolidation.}\ 
Define
\[ B_{(2,-1)} = \dfrac{8}{3} d_{\max}^2 M,\quad B_{(2)} = 2 d_{\max}^2, \quad
      B_{(-1)}= \nfrac{1}{2} M^2. \]
It follows from (\ref{B21})--(\ref{B1}) that
\begin{align*}
 & \frac{|\mathcal{L}^\ast(Z)|}{|\Omega(\boldsymbol{d})|} \\
  &=
  \sum_{p=0}^2\,\sum_{q=0}^3\, \frac{|\mathcal{C}(p,q)|}{|\mathcal{C}(0,0)|} \\
   &\leq
  1 + B_{(2)} + B_{(2)}^2 + B_{(-1)} + B_{(-1)}B_{(2)} + B_{(-1)} B_{(2)}^2\\
  & \hspace*{8mm} {} 
    B_{(-1)} B_{(2)} B_{(2,-1)}
    + B_{(-1)}^2  + B_{(-1)}^2  B_{(2)} 
   \\ & \hspace*{8mm} {}  
   + B_{(-1)}^2\, B_{(2,-1)} + B_{(-1)}^3 \\
 &\leq \nfrac{1}{5} M^6,
\end{align*}
using the upper bound on $d_{\max}$ and the fact that $M\geq 144$.
This completes the proof of Lemma~\ref{3switches}. \hfill \qedbox
\end{proof}

Since $M\leq d_{\max} n$, the bound $\nfrac{1}{5} M^6$ is at most
a factor $n/10$ bigger than the analogous bound $2d^6 n^5$ given
in~\cite[Lemma 4]{CDG} in the regular case.

Finally we can prove Theorem~\ref{main}.

\begin{proof} (\emph{Proof of Theorem~\ref{main}})\
We wish to apply (\ref{flowbound}).
It follows from the configuration model (see~\cite[Equation (1)]{McKW91}) that
the set $\Omega(\boldsymbol{d})$ has size
\begin{equation} |\Omega(\boldsymbol{d})| 
    \leq \frac{M!}{2^{M/2}\, (M/2)!\, \prod_{j=1}^n d_j!} \leq
          \exp\left( \nfrac{1}{2} \, M\log(M) \right).
\label{size}
\end{equation}
Hence the smallest stationary probability $\pi^\ast$ satisfies
$\log(1/\pi^\ast) = \log(|\Omega(\boldsymbol{d})|) \leq M\log(M)$. 
Next, $\ell(f)\leq M/2$ since each transition along a canonical path
replaces an edge of $G$ by an edge of $G'$.

Finally, if $e=(Z,Z')$ is a transition of the switch chain then
$1/Q(e) = 6\, a(\boldsymbol{d}) \leq M^2$, using (\ref{ad}).
Combining this with Lemmas~\ref{fbound} and~\ref{3switches} gives
$\rho(f) \leq \nfrac{1}{5}  d_{\max}^{14}\, M^8$.
Substituting these expressions into (\ref{flowbound}) gives
the claimed bound on the mixing time.  \hfill \qedbox
\end{proof}

\end{document}